\documentclass{article}

\usepackage{amsmath}
\usepackage{amsthm}
\usepackage[pdftex]{hyperref}
\newtheorem{theorem}{Theorem}
\newtheorem{lemma}[theorem]{Lemma}
\newtheorem{proposition}[theorem]{Proposition}

\def\keywords#1{\par\vspace*{8pt}	
{\footnotesize{\noindent{\it Keywords}\/:\ #1\par}}\par}

\title{On Direct Product and Quotients of \\ Strongly Connected Automata}

\author{Zino H. Hu
\footnote{The author is a retiree of The Boeing Company. He is not affiliated with any organization at present.}\\
\href{mailto:zinohu@gmail.com}{zinohu@gmail.com}}

\begin{document}

\maketitle

\begin{abstract}
Let $A \ \times \ B$ be the direct product of a strongly connected permutation automaton $A$ and a strongly connected synchronizing (reset) automaton $B$, then $A \ \times \ B$ is strongly connected and

\boldmath\[A \cong (A \times B)/\pi\] \[B \cong (A \times B)/\rho\] \[(A \times B) \  \cong \  (A \times B)/\pi \  \times \ (A \times B)/\rho\]\unboldmath where $\pi$ and $\rho$ are automaton congruence relations defined in this paper, $(A \times B)/\pi$ and $(A \times B)/\rho$ are quotient automata constructed by $\pi$ and $\rho$ respectively.

\keywords{automata; strongly connected; permutation; synchronizing; reset; direct product; automata congruence; quotient; minimal ideal; right group}

\end{abstract}

\section{Introduction}	

When considering the direct product of automata, one would intuitively think that the direct product of two strongly connected automata would be strongly connected. Fleck \cite{Fleck} showed that for any non-trivial (more than one state) automata $A$ and $B$, if $A$ is homomorphic to $B$ (or vice versa) then $A \times B$ is not strongly connected. This is because an automaton homomorphism is a transition preserving function. Any input string which sends a state $s$ to $s$ itself in $A$ must send every state $t$ to $t$ itself in $B$ if there is a homomorphism from $A$ to $B$, therefore there is no transition between the states $(s,t_1)$ and $(s,t_2)$ in $A \times B$ when $t_1 \neq t_2$. Thus, if the direct product of two strongly connected automata is strongly connected then there is no homomorphism between them.

When the transition semigroup of an automaton is a group, we say the automaton is a permutation automaton. When an automaton has at least one reset input function, i.e. an input function which maps every state into a single fixed state, we say the automaton is a synchronizing (reset) automaton. Bavel et al.\cite{BavelPaper}(Theorem 7.5) showed that the direct product of a permutation strongly connected automaton and a reset strongly connected automaton is strongly connected. 

If a strongly connected automaton can be decomposed into two quotient automata by automata congruence relations, then the quotient automata obtained by decomposition must be strongly connected because of the canonical homomorphisms. In that case, is it isomorphic to the direct product of its quotient automata? Specifically, if a strongly connected automaton can be decomposed into a permutation strongly connected quotient automaton and a synchronizing strongly connected quotient automaton, is it isomorphic to the direct product of these two quotient automata? 

Masunaga et al. \cite{MasunagaPaper} (Corollary 5.6) showed that a strongly connected state-independent automaton is isomorphic to a direct product of a strongly connected state-independent permutation automaton and a strongly connected reset automaton.

In this paper, we answer the question: what exactly is the direct product of a  strongly connected permutation  automaton (not necessarily state-independent) and a strongly connected  synchronizing (reset) automaton?
 
\section{Preliminaries}
The definition of an automaton in this paper is from Bavel \cite {BavelBook}. For a non-empty finite set $\Sigma$, we denote the free monoid over $\Sigma$ by $\Sigma^*$ and the empty input string by $\varepsilon$. An automaton\ \footnote{It is often called semiautomaton in other literature.} is a triple $A=(S,\ \Sigma,\ \delta)$ where $S$ is a set of states; $\Sigma$ is a non-empty set called the input alphabet;  $\delta : S \times \Sigma^* \to S$ is the transition function satisfying $\forall s \in S, \ \forall x, y \in \Sigma^*,\ \delta(s, xy) = \delta(\delta(s, x), y)$ and $\delta(s, \varepsilon) = s$.

An automaton is finite if and only if its set of states is finite. All automata are finite in this paper.

For the sake of brevity, $A = (S, \Sigma,\delta)$ and $B = (T,\Sigma,\gamma)$ are finite non-trivial automata throughout this paper.

An automaton $A$ is strongly connected(transitive) if for every $s, t \in S$, there exists $x\in\Sigma$* such that $\delta(s, x) = t$.

A mapping $\alpha : S \to T$ such that $\alpha(\delta(s, x)) = \gamma(\alpha(s), x)$ for every $s \in S$, $x \in \Sigma ^*$ is a homomorphism from $A$ to $B$. An isomorphism is a bijective homomorphism. If there is an isomorphism between $A$ and $B$, we say $A$ is isomorphic to $B$, denoted by $A \cong B$.

Note that any homomorphism between two strongly connected automata is surjective.

Define a relation $\equiv_A$ on $\Sigma ^*$ by $\forall x, y \in \Sigma^*, x\equiv _A y$ if and only if $\delta(s, x) = \delta(s, y)$ for every $s \in S$. The relation $\equiv _A$ is an equivalence relation. Let $x \in \Sigma ^*$, we denote the equivalence class $\{ y \in \Sigma ^* : x \equiv _A y \} $ by $ [x]_A$. 

We denote $\{[x]_A : x \in \Sigma ^* \}$ by $M_A$. We define an operation on $M_A$ by $[x]_A[y]_A = [xy]_A$ and $[\varepsilon]_A[x]_A = [x]_A[\varepsilon]_A  = [x]_A$ for every $[x]_A, [y]_A \in M_A$. $M_A$ with this operation is a monoid with the identity $[\varepsilon]_A$ and is called the transition monoid of $A$. 

We denote $\{[x]_A : x \in \Sigma ^+ \}$ by $S_A$ where $\Sigma ^+ = \Sigma ^* - \{\varepsilon\}$. We define an operation on $S_A$ by $[x]_A[y]_A = [xy]_A$ for every $[x]_A, [y]_A \in S_A$. $S_A$ with this operation is a semigroup and is called the transition semigroup of $A$.

When viewing $x \in \Sigma^*$ as a function from S to S, i.e. $x : S \to S$,  $x(s) = \delta(s, x)$ for every $s \in S$ we call $x$ an input function. We denote the range of an input function $x$ by $Im_A(x) = \{\delta(s, x) : s \in S\}$ and the rank of $x$ by $|Im_A(x)|$.

A right ideal $R$ of a semigroup $S$ is a nonempty subset of $S$ such that $r \in R$ and $s \in S$ imply $rs \in R$. This is equivalent to $RS \subseteq R$. Similarly a left ideal of $S$ is a nonempty subset $L$ satisfying $SL \subseteq L$. An ideal $I$ of a semigroup $S$ is a subset which is both a left and a right ideal. It satisfies $IS \cup SI \subseteq I$. An ideal $I$ is called a minimal ideal if for every ideal $J$ of $S$, $J \subseteq I$ implies $J = I$.

An idempotent element of a semigroup $S$ is an element $e \in S$ such that $ee = e$. Any finite semigroup has at least one idempotent element.

Let $I_A =$ \{$[x]_A \in S_A : |Im_A(x)|$ is minimal\} be the minimal ideal of the transition semigroup $S_A$ (by \cite{LallementBook} Chapter 2 Proposition 1.3). $I_A$ is called the minimal transition ideal of $A$. $I_A$ is a finite semigroup. Thus it has at least one idempotent element. Denote the set of all the idempotent elements in $I_A$ by $E_A$. Note that $E_A \subseteq I_A$.

An automaton $A$ is called a permutation automaton if every input is a permutation, i.e. every input function is bijective. $A$ is a permutation automaton if and only if its transition monoid $M_A$ is a group if and only if $\delta(s, x) = \delta(t, x)$ implies $s = t$ for all $s, t \in S$, $x \in \Sigma^*$ (since $A$ is finite). It is possible that the transition semigroup $S_A$ of an automaton $A$ is a group even when $A$ is not a permutation automaton. However, if an automaton $A$ is strongly connected and its transition semigroup $S_A$ is a group then its transition monoid $M_A$ is a group, i.e. $A$ is a permutation automaton.

Let $x \in \Sigma^*$, $x$ is a reset input function of A if there exists $t \in S$ such that $\delta(s, x) = t$ for all $s \in S$. If $x$ is a reset input function of $A$ then $[x]_A$ is a right zero element of the transition semigroup $S_A$. Thus $zx \equiv _A x$ and $xz$ is a reset input function of A for every $z \in \Sigma^*$ if $x$ is a reset input function of $A$. An automaton $A$ is called a synchronizing (reset) automaton if $A$ has at least one reset input function. The minimal ideal $I_A$ of a synchronizing automaton $A$ consists of all reset input functions of $A$, hence a right zero semigroup (\cite{CliffordBook} Chapter 1.1 p.4).

The direct product of an automaton $A$ and an automaton $B$ studied in this paper is the automaton $A \times B = (S \times T, \ \Sigma, \ \delta_{A \times B})$ where $\delta_{A \times B}((s, t), x) = (\delta(s, x), \gamma(t, x))$ for every $s \in S$, $t \in T$, $x \in \Sigma ^*$. Note that $A$ and $B$ have the same input alphabet $\Sigma$ in the definition. This type of direct product is called homogeneous direct product. Dorfler  studied inhomogeneous direct product of strongly connected automata in \cite{DorflerPaper}.

A semigroup $S$ is right simple if $S$ itself is the only right ideal of $S$. A semigroup $S$ is right simple if and only if $aS = S$ for every $a \in S$ if and only if $\forall \ a, b \in S$, there exists $x \in S$ such that $ax = b$ (\cite{CliffordBook} Chapter 1.1 p.6). A semigroup is a right group if and only if it is a right simple semigroup and has an idempotent element(\cite{CliffordBook} Chapter 1.11 Theorem 1.27). Thus a finite semigroup is a right group if and only if it is right simple.

An automaton $A$ is called a quasi-ideal automaton if (1) it is strongly connected; (2) the minimal ideal $I_A$ of its transition semigroup is a right group (3) the ranges of the idempotent elements of the minimal ideal $I_A$ of its transition semigroup form a partition on S, i.e. $\cup_{[e]_A \in E_A}Im_A(e) = S$ and $\forall \ [e]_A, [f]_A \in E_A$, $Im_A(e) \cap Im_A(f) = \emptyset$ or $Im_A(e) = Im_A(f)$.

\section{Quasi-Ideal Automata}

A synchronizing (reset) strongly connected automaton has a property that the cardinality of the minimal transition ideal equals the number of states of the automaton.

\begin{lemma}\label{resetsclemma}
Let $A$ be a strongly connected automaton. If $A$ has at least one reset input function then there exists a unique $($up to the equivalence relation $\equiv_A)$ reset input function $x_s$ for every $s \in S$  such that $\delta(t, x_s) = s$ for all $t \in S$.\end{lemma}

\begin{proof}
Let $x$ be a reset input function of $A$. By definition, there exists $s_0 \in S$ such that $\delta(t, x) = s_0$ for all $t \in S$. Let $s \in S$. Since $A$ is strongly connected, there exists $z \in \Sigma^*$ such that $\delta(s_0, z) = s$. Let $x_s = xz$. Then $\delta(t, x_s) = \delta(t, xz) =  \delta(s_0, z) = s$ for all $t \in S$. Thus, $x_s$ is a reset input function such that $\delta(t,x_s) = s$ for all $t \in S$. If $x^\prime_s$ is another reset input function such that  $\delta(t, x^\prime_s) = s$ for all $t \in S$, then $x_s  \equiv_A x^\prime_s$.
\end{proof}

The following theorem is due to Bavel et al.\cite{BavelPaper}(Theorem 7.5). A different proof is provided here.

\begin{theorem} \label{theoremhong}
The direct product of a permutation strongly connected automaton and a reset strongly connected automaton is strongly connected.
\end{theorem}

\begin{proof}
Let $A$  be a permutation strongly connected automaton and $B$ be a synchronizing strongly connected automaton. Let $(s_1, t_1), (s_2, t_2) \in S \times T$. Since $A$ is strongly connected, there exists $x \in \Sigma^*$ such that $\delta(s_1, x) = s_2$. By Lemma \ref {resetsclemma}, there exists a reset input function $y \in \Sigma^*$ such that  $\gamma(t, y) = t_2$ for all $t \in T$. In particular, $\gamma(t_1, y) = t_2$. $S_A$ is a group since $A$ is a permutation automaton. Thus there exists $z \in \Sigma^+$ such that $zy \equiv _A \varepsilon$. Consider the input string $xzy \in \Sigma^*$. $xzy = x(zy) \equiv _A x\varepsilon \equiv_A x$. On the other hand, $xzy \equiv _B y$ since $y$ is a reset input function of $B$. Now \[\delta_{A \times B}((s_1, t_1), xzy) = (\delta(s_1, xzy), \gamma(t_1, xzy)) = (\delta(s_1, x), \gamma(t_1, y)) = (s_2, t_2)\] Hence $A \times B$ is strongly connected.
\end{proof}

The next lemma is useful for analyzing minimal rank input functions of the direct product of automata.

\begin{lemma}\label{minimalideallemma}
Let $A$ and $B$ be automata. Let $x \in \Sigma^*$. Then 

\begin{itemize}
\item[] \ \rm({i})\ \  $Im_{A \times B}(x) = Im_A(x) \times Im_B(x)$.
\item[] \ \rm({ii}) \ $[x]_{A \times B} \in I_{A \times B}$ if and only if $ [x]_A \in I_A$ and $[x]_B \in I_B$.
\end {itemize}

\end{lemma}

\begin{proof}
(i) is trivial. By (i), $|Im_{A\times B}(x)|$ is minimal if and only if both $|Im_A(x)|$ and $|Im_B(x)|$ are minimal.
\end{proof}

If $R$ is a right group and $E$ is the set of all the idempotents in $R$, then $R = \cup_{e \in E}Re$ (\cite{LallementBook} Chapter 1 Exercises 6(b)). Thus $I_A = \cup_{[e]_A \in E_A}I_A[e]_A$ when $I_A$ is a right group.

\begin{lemma}\label{rightgroupunionlemma}
Let $A$ be a strongly connected automaton. If $I_A$ is a right group, then $S = \cup_{e \in E_A} Im_A(e)$.
\end{lemma}

\begin{proof}
$I_A$ itself is its only right ideal. By \cite{LallementBook} Chapter 8 Corollary 2.6, $S = \cup _{[x]_A \in I_A}Im_A(x)$. Since $Im_A(xe) \subseteq Im_A(e)$ for all $[x]_A \in I_A$ and $[e]_A \in E_A$, we have $S = \cup_{[xe]_A \in I_AE_A}Im_A(xe) \subseteq \cup_{[e]_A \in E_A}Im_A(e)\subseteq S$ 
\end{proof}

\begin{theorem}\label{quasiidealtheorem}
The direct product of a permutation strongly connected automaton and a synchronizing strongly connected automaton is a quasi-ideal automaton.
\end{theorem}

\begin{proof}
Let $A$ be a permutation strongly connected automaton and $B$ be a synchronizing strongly connected automaton. By Theorem \ref{theoremhong}, $A \times  B$ is strongly connected. Let $[x]_{A \times B}$, $[y]_{A\times B} \in I_{A \times B}$. By Lemma \ref{minimalideallemma} (ii), $[x]_A, [y]_A \in I_A$ and $[x]_B, [y]_B \in I_B$. Since $A$ is a permutation automaton, $M_A$ is a group. Therefore, there exists $x^\prime \in \Sigma^*$ such that $xx^\prime \equiv _A \varepsilon$. Then $xx^\prime y \equiv_A y$. Since $[y]_B \in I_B$, $y$ is a reset input function of $B$. We have $xx^\prime y \equiv_B y$. Thus $xx^\prime y \equiv_{A \times B} y$, i.e. the minimal transition ideal $I_{A \times  B}$ of $A \times B$ is right simple hence a right group. \par 
By Lemma \ref{rightgroupunionlemma}, $S \times T  = \cup _{[e]_{A \times B} \in E_{A \times B}} Im_{A \times B}(e)$. Suppose $Im_{A \times B}(e) \cap Im_{A \times B}(f) \neq \emptyset$ for some $[e]_{A \times B}, [f]_{A \times B} \in E_{A \times B}$. By Lemma \ref{minimalideallemma} (i), we have $Im_A(e) \cap Im_A(f)\neq \emptyset$ and $Im_B(e) \cap Im_B(f) \neq \emptyset$. Both $e$ and $f$ are reset input functions of $B$, $|Im_B(e)| = |Im_B(f)| = 1$. So $Im_B(e) = Im_B(f)$. $M_A$ is a group, $Im_A(e) = Im_A(f) = S$. Hence $ Im_{A \times B}(e) = Im_A(e) \times Im_B(e) = Im_A(f) \times Im_B(f) = Im_{A \times B}(f)$. Therefore $\{Im_{A \times B} (e) : [e]_A \in E_{A \times B}\}$ forms a partition on $S \times T$.
\end{proof}

We have shown that we can produce a quasi-ideal automaton by taking the direct product of a permutation strongly connected automaton and a synchronizing strongly connected automaton. We shall decompose a quasi-ideal automaton by using automaton congruence relations.

First, we shall prove that if the minimal ideal of the input semigroup of an automaton is a right group, then the automaton is state independent with respect to the minimal ideal of the input semigroup.

\begin{proposition}\label{rightgroupSIproposition}
Let $A$ be an automaton. If the minimal ideal $I_A$ of its input semigroup is a right group, then $\forall s, t \in S$, $\delta(s, x) = \delta(t, x)$ for some $x \in I_A \implies \delta(s, y) = \delta(t, y)$ for every $y \in I_A$.
\end{proposition}
\begin{proof}
Let $s, t \in S$. Suppose $\delta(s, x) = \delta(t, x)$ for some $x \in I_A$, let $y \in I_A$. Then $y \equiv xz$ for some $z \in I_A$ since $I_A$ is a right group. Now, \[\delta(s, x) = \delta(t, x) \implies \delta(s, xz) = \delta(t, xz)
\implies \delta(s, y) = \delta(t, y)\]
\end{proof}
An automaton congruence relation on an automaton $A$ is an equivalence relation $\theta$ on $S$ compatible with the transition function, i.e. $s \ \theta \ t$ implies $\delta(s, z) \ \theta \ \delta(t, z)$ for all $z \in \Sigma ^*$. We denote the $\theta$ equivalence class of $s\in S$ by $[s]_\theta = \{t \in S : t \ \theta \ s\}$ and $S/\theta = \{[s]_\theta : s \in S\}$. With an automaton congruence relation $\theta$ on $A$, we can construct the $\theta$-quotient automaton $A/\theta = (S/\theta, \Sigma, \delta_{A/\theta})$ where $\delta_{A/\theta} : S/\theta \times \Sigma^* \to S/\theta$ is defined by $\delta_{A/\theta}([s]_\theta, z) = [\delta(s, z)]_\theta$ for all $z \in \Sigma^*$.

Note that there is a surjective canonical homomorphism from an automaton $A$ to its quotient automaton $A/\theta$. If $A$ is strongly connected, so is $A/\theta$.

\begin{proposition}\label{pipropertyproposition}
Let $A$ be a strongly connected automaton. If the minimal ideal of the input semigroup of $A$ is a right group, then the relation $\pi$ on $S$ defined by $\forall s, t \in S, s \ \pi \ t$ if and only if $\delta(s, x) = \delta(t, x)$ for some $x \in I_A$ is an automaton congruence. Moreover, $\forall s, t \in S, \forall z \in \Sigma^*, \delta(s, z) \ \pi \ \delta(t, z) \implies s \ \pi \ t$.
\end{proposition}
\begin{proof}
We need to show $\pi$ is transitive. Let $r, s, t \in S$. Let $r \ \pi \ s$ and $s \ \pi \ t$.  Then $\delta(r, x) = \delta(s, x)$ for some $x \in I_A$ and $\delta(s, y) = \delta(t, y)$ for some $y \in I_A$. By assumption, $I_A$ is a right group, by Proposition \ref{rightgroupSIproposition} , $\delta(r, y) = \delta(s, y)$. Then $\delta(r, y) = \delta(t, y)$, that is $r \ \pi \ t$. $\pi$ is transitive.
Let $s, t \in S, z \in \Sigma^*$. Suppose $s \ \pi \ t$, then $\delta(s, x) = \delta(t, x)$ for some $x \in I_A$. $zx \in I_A$ since $I_A$ is an ideal. By Proposition  \ref{rightgroupSIproposition}  again, $\delta(\delta(s, z), x) = \delta(s, zx) = \delta(t, zx) = \delta(\delta(t, z), x)$, that is, $\delta(s, z) \ \pi \ \delta(t, z)$. $\pi$ is an automaton congruence.

Suppose $\delta(s, z) \ \pi \ \delta(t, z)$, that is $\delta(\delta(s, z), x) = \delta(\delta(t, z), x)$ for some $x \in I_A$. Hence $\delta(s, zx) = \delta(t, zx)$. $zx \in I_A$ since $I_A$ is an ideal. $s \ \pi \ t$ by definition of $\pi$.
\end{proof}

The idea of the automaton congruence relation $\pi$ is from Perrot \cite{PerrotPrefixPaper} (Theorem 4), Perrot \cite {PerrotGroupsPaper}(Theorem 7) and Lallement \cite {LallementBook} (Chapter 8 Proposition 4.4).

\begin{theorem}\label{theorempi}
Let $A$ be a strongly connected automaton. If the minimal ideal of the input semigroup of $A$ is a right group, then there exists an automaton congruence relation $\pi$ on S defined by $\forall$ $s, t \in S$, $s \ \pi  \ t$ if and only if $\delta(s, x) = \delta(t, x)$ for some $[x]_A \in I_A$. The $\pi$-quotient automaton $A/\pi$ is a strongly connected permuation automaton.
\end{theorem}
\begin{proof}
By Proposition \ref{pipropertyproposition} , the automaton congruence relation $\pi$ exists.  Let $[s]_\pi, [t]_\pi \in S/\pi$. Since $A$ is strongly connected, $\exists z \in \Sigma^*$ such that $\delta(s, z) = t$. Then $\delta_{A/\pi}([s]_\pi, z) = [\delta(s, z)]_\pi = [t]_\pi$. Thus $A/\pi$ is strongly connected. Let $z \in \Sigma^*, [s]_\pi, [t]_\pi \in S/\pi$ and $\delta_{A/\pi}([s]_\pi, z) = \delta_{A/\pi}([t]_\pi, z)$. Then $[\delta(s,z)]_\pi = [\delta(t, z)]_\pi$, i.e. $\delta(s, z) \ \pi \ \delta(t, z)$. By Proposition \ref{pipropertyproposition} , $s \ \pi \ t$, i.e. $[s]_\pi = [t]_\pi$. Thus, $z$ is an injective input function from $S/\pi$ to $S/\pi$. $S/\pi$ is finite. $z$ is a permutation. $A/\pi$ is a permutation automaton. $A/\pi$ is strongly connected since $A$ is strongly connected.
\end{proof}

\begin{proposition}\label{rhodefproposition}
Let $A$ be an automaton. If the ranges of the elements of the minimal ideal of the input semigroup of $A$ form a partition on $S$, then there is an automaton congruence $\rho$ on $S$ such that $\forall s, t \in S, s \ \rho \ t$ iff $s \in Im(x)$ and $t \in Im(x)$ for some $x \in I_A$.
\end{proposition}
\begin{proof}
By assumption, $\{Im_A(x) : [x]_A \in I_A\}$ forms a partition on $S$. Let $\rho$ be the equivalence relation on $S$ induced by the partition, i.e. for every $s, t \in S$, $s \ \rho\  t$ if and only if $s, t \in Im_A(x)$ for some $[x]_A \in I_A$.
 
To show $\rho$ is an automaton congruence relation, let $s, \ t \in S$, $z \in \Sigma^*$. Suppose $s \ \rho \ t$, then there exists $[x]_A \in I_A$ and $s^\prime, \ t^\prime \in S$ such that $s = \delta(s^\prime, x)$ and $t = \delta(t^\prime, x)$. $\delta(s, z) = \delta(s^\prime, xz)$ and $\delta(t, z) = \delta(t^\prime, xz)$.   Hence $\delta(s, z) \in Im_A(xz)$ and $\delta(t, z) \in Im(_A(xz)$.  Since $I_A$ is an ideal, $[xz]_A \in I_A$. Thus, $\delta(s, z) \ \rho \ \delta(t, z)$. $\rho$ is an automaton congruence.
\end{proof}

The idea of the automaton congruence relation $\rho$ is from Lallement \cite {LallementBook} (Chapter 8 Theorem 4.8).

\begin{theorem}\label{theoremrho}
Let $A$ be a strongly connected automaton. If there exists an automaton congruence relation $\rho$ on $S$ defined by $\forall$ $s, t \in S$, $s \ \rho \ t$ if and only if $s , t \in Im_A(x)$ for some $[x]_A \in I_A$, the $\rho$-quotient automaton $A/\rho$ is a synchronizing strongly connected automaton.
\end{theorem}
\begin{proof}
By Proposition \ref{rhodefproposition} , there exists an automaton congruence relation $\rho$ on $S$ defined by $\forall$ $s, t \in S$, $s \ \rho \ t$ if and only if $s , t \in Im_A(x)$ for some $[x]_A \in I_A$. $A/\rho$ is strongly connected since $A$ is strongly connected. Let $[x]_A \in I_A, \ s \in Im_A(x)$. Consider $[s]_\rho$. Let $[t]_\rho \in S/\rho$. Since $\delta(t, x) \in Im_A(x)$, we have $s \ \rho \ \delta(t, x)$, i.e. $[s]_\rho = [\delta(t, x)]_\rho = \delta_{A/\rho}([t]_\rho, x)$. Since $[t]_\rho$ was arbitrarily chosen, $x$ is a reset input function of $A/\rho$. Thus, $A/\rho$ is a synchronizing automaton.
\end{proof}

We have shown that we can decompose a quasi-ideal automaton into two strongly connected quotient automata, one is permuting and the other is synchronizing. We shall show that a quasi-ideal automaton is isomorphic to the direct product of its quotient strongly connected permutation automaton and its quotient strongly connected synchronizing automaton.

Define a binary operation $\circ$ on relations $\pi$ and $\rho$ on a set $S$ by $\pi \circ \rho = \{(s,t) \in S \times S : \exists \ u \in S$ such that $(s,u) \in \pi$ and $(u,t) \in \rho\}$. We denote the equality (identity) relation $\{(s,s) : s \in S\}$ by $1_S$.

The following theorem is from Masunaga et al. \cite {MasunagaPaper}  (Sec.5). Detailed proof is provided here.
\begin{theorem}\label{MasunagaReference}
Let $A$ be an automaton. If there exist automaton congruence relations $\pi$ and $\rho$ on $A$ such that $\pi \cap \rho = 1_S$ and $\pi \circ \rho = S \times S$ then $A \ \cong \ A/\pi \times A/\rho$.
\end{theorem}

\begin{proof}
Define $\alpha : S \to S/\pi \times S/\rho$ by $\alpha(s) = ([s]_\pi,[s]_\rho)$ for all $s \in S$. $\alpha$ is well defined and a homomorphism. Let $s, t \in S$ and $\alpha(s) = \alpha(t)$, i.e. $([s]_\pi,[s]_\rho) = ([t]_\pi, [t]_\rho)$. Then $s \ \pi \ t$ and $s \ \rho \ t$, i.e. $(s,t) \in \pi \cap \rho = 1_S$. We have $s = t$. Hence $\alpha$ is injective. To show that $\alpha$ is surjective, let $([s]_\pi, [t]_\rho) \in S/\pi \times S/\rho$. Since $(s,t) \in S \times S = \pi \circ \rho$, there exists $u \in S$ such that $(s,u) \in \pi$ and $(u,t) \in \rho$. Thus, $[s]_\pi = [u]_\pi$ and $[u]_\rho = [t]_\rho$. We have $\alpha(u) = ([u]_\pi,[u]_\rho) = ([s]_\pi,[t]_\rho)$.
\end {proof}

\begin{theorem}\label{maintheorem}
An  automaton is isomorphic to the direct product of a permutation strongly connected automaton and a synchronizing strongly connected automaton if and only if it is a quasi-ideal automaton.
\end{theorem}
\begin{proof}
Theorem \ref{quasiidealtheorem} establishes the only if part. 

Let $A$ be a quasi-ideal automaton. By Theorem \ref{theorempi}, $A/\pi$ is a permutation strongly connected automaton where $\pi$ is the automaton congruence relation defined on $S$ by $\forall \ s, \ t \in S, \ s \ \pi  \ t$ if and only if $\delta(s, x) = \delta(t, x)$ for some $[x]_A \in I_A$. By Proposition \ref{rightgroupSIproposition} , $\forall \ s, \ t \in S$, $s \ \pi  \ t \implies \delta(s, x) = \delta(t, x)$ for every $[x]_A \in I_A$.

By Theorem \ref{theoremrho}, $A/\rho$ is a synchronizing strongly connected automaton where $\rho$ is the automaton congruence relation defined on $S$ by $\forall \ s, \ t \in S, \ s \ \rho \ t$ if and only if $s, \ t \in Im_A(e)$ for some $[e]_A \in E_A$. 

Let $s, \ t \in S$. Suppose $(s,t) \in \pi \cap \rho$. Then $s, \ t \in Im_A(e)$ for some $e \in E_A$ since $s \ \rho \ t$. There exist $s^\prime, t^\prime \in S$ such that $s = \delta(s^\prime, e)$ and $t = \delta(t^\prime, e)$. We have $\delta(s, e) = \delta(t, e)$ since $s \ \pi \ t$. Now \[s = \delta(s^\prime, e) = \delta(s^\prime, ee) = \delta(s, e) = \delta(t, e) = \delta(t^\prime, ee) = \delta(t^\prime, e) = t\] Hence $(s,t) \in 1_S$ i.e. $\pi \ \cap \ \rho = 1_S$. 

Let $(s,t) \in S \times S$. By Lemma \ref{rightgroupunionlemma} , $t \in Im_A(e)$ for some $[e]_A \in E_A$. Let $[x]_A \in I_A$. Since $I_A$ is right simple, $ez \equiv_A x$ for some $z \in I_A$. \[\delta(s, x) = \delta(s, ez) = \delta(s, eez) = \delta(\delta(s,e), x)\] So $s \ \pi \ \delta(s, e)$. Since $\delta(s, e) \in Im_A(e)$, $\delta(s, e) \ \rho \ t$. Hence $(s,t) \in \pi \circ \rho$

We have $\pi \circ \rho = S \times S$. By  Theorem \ref{MasunagaReference} , $A \ \cong \ A/\pi \times A/\rho$.
\end{proof}

Let $A$ be a permutation strongly connected automaton. Let $B$ be a syncronizing strongly connected automaton. 
By Theorem \ref{quasiidealtheorem} , $A \times B$ is a quasi-ideal automaton. Thus, $I_{A \times B}$ is a right group and the ranges of the idempotent elements of $I_{A \times B}$ form a partition on $S \times T$. 

By Proposition \ref{pipropertyproposition}, there exists an automaton congruence relation $\pi$ on $S \times T$ defined by $\forall (s_1, t_1), (s_2, t_2) \in S \times T, (s_1, t_1) \ \pi \ (s_2, t_2)$ if and only if $\delta_{A \times B}((s_1, t_1), x) = \delta_{A \times B}((s_2, t_2), x)$ for some $x \in I_{A \times B}$.

By Proposition \ref{rhodefproposition}, there exists an automaton congruence relation $\rho$ on $S \times T$ defined by $\forall(s_1, t_1), (s_2, t_2) \in S \times T$, $(s_1, t_1) \ \rho \ (s_2, t_2)$ if and only if $(s_1, t_1), (s_2, t_2) \in Im_{A \times B}(x)$ for some $[x]_{A \times B} \in I_{A \times B}$.

\begin {theorem}\label{butterflytheorem}
Let $A \times B$ be the direct product of a strongly connected permutation automaton $A$ and a strongly connected synchronizing automaton $B$, then $A \cong (A \times B)/\pi$, $B \cong (A \times B)/\rho$ and $(A \times B) \  \cong \  (A \times B)/\pi \  \times \ (A \times B)/\rho$ where $\pi$ and $\rho$ are automaton congruence relations on $S \times T$ defined as above, $(A \times B)/\pi$ and $(A \times B)/\rho$ are quotient automata constructed by $\pi$ and $\rho$ respectively.
\end{theorem}

\begin{proof}
Let $s \in S,\  t \in T$. Then $A = <s>$, $B = <t>$, $(A \times B)/\pi = <[(s, t)]_\pi>$ and $(A \times B)/\rho = <[(s, t)]_\rho>$   because $A$, $B$, $(A \times B)/\pi$ and $(A \times B)/\rho$ are strongly connected.

Define $\alpha : (A \times B)/\pi \to A$ by $ \alpha(\delta_{A \times B}([(s, t)]_\pi, x) = \delta(s, x), \forall x \in \Sigma^*$.
Let $x_1, x_2 \in \Sigma^*$ and $\delta_{A \times B}([(s, t)]_\pi, x_1) = \delta_{A \times B}([(s, t)]_\pi, x_2)$. 

Then $\delta_{A \times B}((s, t), x_1)\ \pi \  \delta_{A \times B}((s, t), x_2)$. $A$ is a permutation automaton, there exists $x_1' \in \Sigma^*$ such that $x_1x_1' \equiv_A \varepsilon$. By Proposition \ref{rightgroupSIproposition}, $\delta(s, x_1x_1') = \delta(s, x_2x_1')$ since $x_1' \in S_A = I_A$ and $\delta(s, x_1)\  \pi \  \delta(s, x_2)$. We have $s = \delta(s, x_2x_1')$. Therefore $\delta(s, x_1) = \delta(s, x_2x_1'x_1) = \delta(s, x_2)$ since $x_1'x_1 \equiv_A \varepsilon$. So, $\alpha$ is well-defined. Let $\delta(s, x_1) = \delta(s, x_2)$. Since $B$ is a synchronizing automaton, there exists a reset input function $[x]_B \in I_B$ such that $\gamma(t, x_1x) = \gamma(t, x) = \gamma(t, x_2x)$. Thus \[\delta_{A \times B}((s, t), x_1) = (\delta(s, x_1), \gamma(t, x_1)) \  \pi \  (\delta(s, x_2), \gamma(t, x_2)) = \delta_{A \times B}((s, t), x_2)\] i.e. $\delta_{A \times B}([(s, t)]_\pi, x_1) = \delta_{A \times B}([(s, t)]_\pi, x_2)$. Hence $\alpha$ is one-to-one. $\alpha$ is surjective because $A$ is strongly connected.

$\alpha([(s, t)]_\pi) = \alpha(\delta_{A \times B}([(s, t)]_\pi, \varepsilon) = \delta(s, \varepsilon) = s$. $\alpha$ is a homomorphism. We have $A \cong (A \times B)/\pi$.

Define $\beta : \  (A \times B)/\rho \to B$ by $\beta(\delta_{A \times B}([(s, t)]_\rho, x) = \gamma (t, x), \forall x \in \Sigma^*$.

Let $x_1, x_2 \in \Sigma^*$ and $\delta_{A \times B}([(s, t)]_\rho, x_1) = \delta_{A \times B}([(s, t)]_\rho, x_2)$. Then by definition of $\rho$, we have $\gamma(t, x_1), \gamma(t, x_2) \in Im(e)$ for some $e \in I_B$. There exist $t_1, t_2 \in T$ such that $\gamma(t, x_1) = \gamma(t_1, e)$ and $\gamma(t, x_2) = \gamma(t_2, e)$. $B$ is a synchronizing automaton, Thus $e$ is a reset input function. Hence $\gamma(t, x_1) = \gamma(t_1, e) = \gamma(t_2, e) = \gamma(t, x_2)$. $\beta$ is well-defined.

Let $\gamma(t, x_1) = \gamma(t, x_2)$. Consider $\delta(s, x_1)$ and $\delta(s, x_2)$. $A$ is a permutation automaton. Both $x_1$ and $x_2$ are permutations. Therefore $\delta(s, x_1), \delta(s, x_2) \in Im(e)$ where $e \in [\varepsilon]_A$. Hence $\delta_{A \times B}([(s, t)]_\rho, x_1) = \delta_{A \times B}([(s, t)]_\rho, x_2)$. Thus $\beta$ is one-to-one. $\beta$ is surjective because $B$ is strongly connected.

$\beta([(s, t)]_\rho) = \beta(\delta_{A \times B}([(s, t)]_\rho, \varepsilon) = \delta(s, \varepsilon) = s$. $\beta$ is a homomorphism. We have $B \cong (A \times B)/\rho$. Thus ($A \times B) \cong (A \times B)/\pi \times (A \times B)/\rho$.  
\end{proof}

The following is an alternative proof of the last part of Theorem \ref{butterflytheorem} :

($A \times B) \cong (A \times B)/\pi \times (A \times B)/\rho$ is a direct consequence of Theorem \ref{maintheorem} because $A \times B$ is a quasi-ideal automaton by Theorem \ref{quasiidealtheorem} .

\pdfbookmark[0]{Acknowledgments}{Ack}
\section*{Acknowledgments}

The author dedicates this paper to the memory of his father who taught him set theory and group theory. The author wants to thank his wife, Debra, for her patience and support during the writing of this paper.

\newpage

\pdfbookmark[0]{References}{Ref}

\end{document}